\documentclass[pra,aps,superscriptaddress,footinbib,twocolumn]{revtex4-2}

\usepackage[dvips]{graphicx} 
\usepackage{amsfonts}
\usepackage{amssymb}
\usepackage{amscd}
\usepackage{amsmath}    
\usepackage{enumerate}
\usepackage{epsfig}
\usepackage{subfigure}
\usepackage{bm}
\usepackage{xcolor}
\usepackage{amsthm}
\usepackage{framed}
\usepackage{multirow}
\usepackage{mathrsfs,amssymb}
\usepackage{latexsym} 
\usepackage{physics}
\usepackage{epstopdf}
\usepackage[skins]{tcolorbox}
\newtcolorbox[auto counter]{tbox}[2][]{%
    enhanced, float=hbt, drop fuzzy shadow southeast,
    colback=white!5!white, colframe=white!50!black,
    width=.97\columnwidth,sharp corners, boxrule=0.8pt,
    title={Table \thetcbcounter: #2}, #1
}

\newtheorem{lemma}{Lemma}
\newtheorem{corollary}{Corollary}

\begin{document}

\title{Continuous-Variable Source-Independent Quantum Random Number Generation with General POVMs}

\author{Hongyi Zhou}
\affiliation{State Key Lab of Processors, Institute of Computing Technology, Chinese Academy of Sciences, 100190, Beijing, China.}

\author{Yu Han}
\affiliation{Henan Key Laboratory of Network Cryptography Technology, Zhengzhou, Henan, 450001, China}

\author{Qiechun Chen}
\affiliation{ZQuantum (Beijing) Technology Co., Ltd, Beijing 100871, China}

\author{Leilei Huang}
\email{huangleilei@zzqtech.com}
\affiliation{ZQuantum (Beijing) Technology Co., Ltd, Beijing 100871, China}

\begin{abstract}
Continuous-variable source-independent quantum random number generators  offer the highest generation rates among semi-device-independent protocols. In reality, the protocol design is limited due to permissible measurement configurations. In this work, we propose a rigorous security proof framework that accommodates general, infinite-dimensional positive-operator-valued measures. Building upon the numerical security proof framework, we evaluate the randomness lower bound by maximizing the eavesdropper's guessing probability. Specifically, we transform the inherently infinite-dimensional semidefinite program in Fock space into a tractable finite-dimensional one, rigorously proving that latter provides a strict upper bound to the guessing probability of the original infinite-dimensional problem. Our framework showcases its capability by certifying secure randomness using unbalanced homodyne detection with only a single quadrature measurement, thereby bypassing the traditional requirement of measuring two conjugate quadratures such as $X$ and $P$. We experimentally validate our protocol on an optical platform using vacuum and weak coherent states, achieving a maximum secure randomness extraction of 1.11 bits per sample and an ultra-high generation rate of 1.776 Gbps. This work provides a flexible design for practical, high-speed quantum random number generators.
\end{abstract}

\maketitle
\section{introduction}
Quantum random number generators (QRNGs) harness the inherent probabilistic nature of quantum mechanics to produce true random numbers, which are essential for cryptographic applications \cite{MaQRNGReview16,RMPQRNGReview16}.
While fully device-independent (DI) QRNGs offer the highest level of security by making no assumptions about the internal workings of the devices, their practical generation rates are extremely low \cite{liu2018device,liu2021device,shalm2021device,PhysRevLett.126.050503}.
On the other hand, trusted-device QRNGs achieve high speeds but are vulnerable to side-channel attacks if the devices are imperfectly characterized \cite{Gabriel10,Nie15,bruynsteen2023integrated}.

To achieve a practical trade-off between uncompromising security and generation rate, semi-device-independent (semi-DI) QRNGs have emerged as a highly promising paradigm \cite{Ma16,PhysRevLett.118.060503,li2019QRNG,PhysRevX.10.041048,avesani2018source,smith2019simple,cao2015loss,Banik15,vsupic2017measurement,bischof2017measurement,nie2024measurement,brask2017megahertz,Brunner15,PhysRevApplied.15.034034,PhysRevA.100.062338,tebyanian2021practical,joch2022certified,ioannou2022steering,zhang2025one}. The semi-DI framework provides robust security guarantees without necessitating demanding Bell-nonlocality tests by introducing justifiable, verifiable physical assumptions on partial components of the system. Various semi-DI architectures have been proposed, including source-independent (SI) QRNGs with trusted  measurement devices and untrusted source \cite{Ma16,PhysRevLett.118.060503,li2019QRNG,PhysRevX.10.041048,avesani2018source,smith2019simple}, measurement-device-independent (MDI) QRNGs with trusted source and untrusted measurement devices \cite{cao2015loss,Banik15,vsupic2017measurement,bischof2017measurement,nie2024measurement}, bounded-energy QRNGs which rely exclusively on restricting the mean photon number of the emitted states \cite{PhysRevA.100.062338,tebyanian2021practical}, and steering-based one-sided device-independent (1SDI) QRNGs where nonlocality is certified assuming only one party's device is characterized \cite{joch2022certified,ioannou2022steering,zhang2025one}. Among these diverse semi-DI categories, SI-QRNGs demonstrate the unique capability of achieving the highest generation rates \cite{avesani2018source,zhang2023realization}.

The development of SI-QRNG protocols represents a critical evolution from discrete-variable (DV) to continuous-variable (CV) regimes. The original source-independent protocol utilized discrete-variable measurements where the security was rigorously established based on the entropic uncertainty relation for conjugate observables \cite{Ma16}. In practice, the operational speeds of such DV protocols were fundamentally bottlenecked by the intrinsic dead time and relatively low detection efficiencies of single-photon detectors. To overcome these constraints, the SI-QRNG framework was subsequently generalized to continuous variables \cite{PhysRevLett.118.060503}. CV measurements employ linear photodetectors that are free from dead-time restrictions, permitting ultra-high sampling bandwidths. Furthermore, CV measurements enables more than 1 bit of genuine randomness per detection event via multi-bit analog-to-digital quantization. These properties enable high-speed randomness generation of CV-SI-QRNGs.

The mathematical core of the security analysis for these advanced SI-QRNGs lies in evaluating the conditional min-entropy, which rigorously bounds the information leakage accessible to an eavesdropper. Analytical security proof methods rely on entropic uncertainty relations of a pair of conjugate observables, such as Pauli $X$ and $Z$ operators, The extremely limited number of such conjugate pairs imposes fundamental constraints on protocol design. Recently, a versatile numerical framework was developed for SI-QRNGs and MDI-QRNGs \cite{PhysRevA.107.052402}. This framework transforms the complex problem of finding the randomness lower bound into a semidefinite programming (SDP) optimization, successfully enabling the automated analysis of various protocols. However, a direct application of this numerical approach to CV-SI-QRNGs encounters a fundamental mathematical obstacle. In CV optical systems, both the adversarial quantum states and the trusted measurement operators are naturally defined in an infinite-dimensional Fock space. Consequently, the SDP formulation becomes intrinsically infinite-dimensional and infeasible.  

To circumvent the infinite-dimensional issue, a natural idea is to make a photon number cutoff assumption, which enforces the infinite-dimensional problem into a truncated finite-dimensional one \cite{lin2019asymptotic}. Another way is to make artificial dimensional reductions, most notably the squashing model \cite{PhysRevLett.101.093601}. This method attempts to construct a virtual completely positive trace-preserving (CPTP) map that mathematically projects the infinite-dimensional Hilbert space into a finite-dimensional subspace (e.g., qubits), ensuring statistical consistency in the measurement outcomes. However, the original squashing model cannot be applied to general protocols (such as the six-state protocol). Subsequent frameworks attempted to broaden this applicability while either resulting in pessimistic key rates \cite{fung2011universal} or introducing more assumptions on measurement operators \cite{li2020improving}. 
Recently, based on the dimension reduction technique \cite{upadhyaya2021dimension}, a partial solution to the numerical security proof problem for CV-SI-QRNGs has been achieved under the phase-insensitive detection assumption \cite{10933507}, which mathematically requires that all positive-operator-valued measures (POVM) elements are diagonal in the Fock basis. This method transforms the intractable infinite-dimensional SDP into a tractable finite-dimensional version through a modified formulation rather than a trivial photon-number cutoff. Crucially, it has been rigorously proven that the optimal value of the resulting finite-dimensional SDP provides a valid upper bound on the guessing probability characterized by the infinite-dimensional matrix linear functional.
However, the general security proof for CV-SI-QRNGs, which allows for arbitrary infinite-dimensional POVMs at the measurement side, remains an outstanding open problem.

In this work, we propose a rigorous security proof framework for CV-SI-QRNGs that accommodates general infinite-dimensional POVMs. By entirely removing the diagonal matrix constraints, our method enables evaluations of cross-interference terms. Combining matrix norm inequalities, we successfully reduce the intractable infinite-dimensional optimization into a solvable finite-dimensional SDP. The latter provides a strict upper bound to the original guessing probability in infinite-dimensional Fock space. Furthermore, we integrate an unbalanced homodyne detection scheme with realistic detector efficiencies into our theoretical model, securing a strictly positive lower bound for randomness. This protocol design requires only single-quadrature measurement. We experimentally validate our protocol on an optical platform using vacuum and weak coherent states, demonstrating an ultra-high secure generation rate of 1.776 Gbps.
\section{Continuous Variable Source-Independent QRNG}
We briefly describe a general CV-SI-QRNG protocol.
First, an untrusted source sends an unknown, infinite-dimensional quantum state to the homodyne detector. Subsequently, an arbitrary measurement is performed where the measurement results are discretized into $m$ intervals, characterized by a POVM set $\{M_j\}_{j=1}^m$. After repeating these state transmission and measurement steps, Alice records the probability distribution of the outcomes, denoted by $p_j$, and proceeds to calculate the upper bound of the guessing probability using a specifically formulated semidefinite programming (SDP) approach.

To quantify the randomness in the asymptotic limit, we first consider the finite-dimensional case. 
The lower bound of randomness is quantified by the guessing probability, 
\begin{equation}
p_{\mathrm{guess}} = \max_{\{\rho_k\}} \sum_{k=1}^m \Tr(\rho_k M_k),
\end{equation}
where $\rho_k$ represents a sub-normalized quantum state $\rho_k =\sum_{i\in S_k}p_i\ket{\psi_i}\bra{\psi_i}$ corresponding to the eavesdropper's knowledge.
To apply numerical methods for computing the upper bound of this guessing probability, one typically sets a photon-number cutoff $N$.
Consequently, the state variables and POVM elements $M_j$ become $N$-dimensional matrices, leading to the following SDP:
\begin{equation}\label{eq:SDPprimalSI}
\begin{aligned}
\max_{\{\rho_k\}} \quad & \sum_{k=1}^m \Tr(\rho_k M_k) \\
\mathrm{s.t.} \quad & \Tr\Big(M_j\sum_{k=1}^m\rho_k\Big) = p_j, \quad \forall j \\
& \Tr\Big(\sum_{k=1}^m \rho_k\Big) = 1 \\
& \rho_k \succeq 0, \quad \forall k,
\end{aligned}
\end{equation}
where the first constraint ensures that the unknown source is compatible with the experimental statistics $p_j$, and the second constraint is the normalization condition.
By solving the SDP problem in Eq.~\eqref{eq:SDPprimalSI}, one obtains the upper bound of the guessing probability $p_{\mathrm{guess}}^U$ in the asymptotic limit.
The extractable randomness is then quantified by the conditional min-entropy:
\begin{equation}\label{eq:randomnessasym}
H_{\mathrm{min}}(A|E) \geq -\log_2 p^U_{\mathrm{guess}}.
\end{equation}

To extend Eq.~\eqref{eq:SDPprimalSI} to CV systems, A primary challenge is that the unknown quantum states, as well as the trusted POVMs, are infinite-dimensional in the Fock basis. Therefore, Eq.~\eqref{eq:SDPprimalSI} for CV systems originally represents an infinite-dimensional SDP, which is computationally intractable. A previous attempt to bypass this was to assume that the detector is phase-insensitive, meaning all POVM elements are diagonal in the Fock basis \cite{10933507}. In this work, we remove this assumption and provide a rigorous randomness lower bound that is valid for general POVMs.

\section{CV-SI-QRNG with General POVM}

\begin{figure}[htp]
\centering
\includegraphics[width=0.5\textwidth]{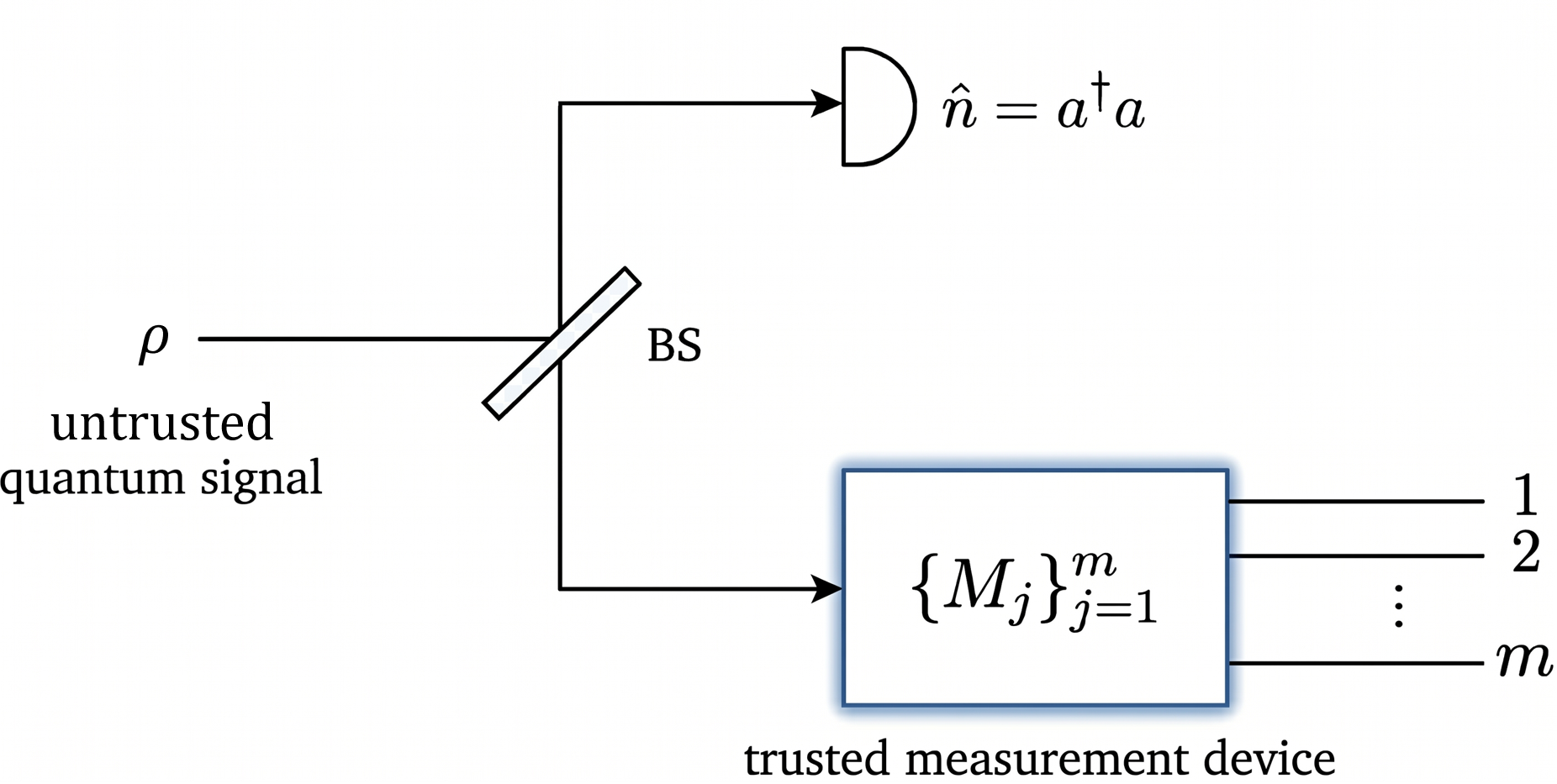}
\caption{Schematic of the general CV-SI-QRNG protocol. The setup consists of an untrusted quantum signal and a fully trusted measurement device. In contrast to discrete-variable schemes, CV schemes require an explicit estimation of the energy components outside the truncated subspace; thus, a beam splitter (BS) diverts a fraction of the signal to a photon number detector ($\hat{n}=a^\dagger a$) to bound the high-energy tail. The trusted measurement device subsequently executes an arbitrary set of POVMs $\{M_j\}_{j=1}^m$. This process inherently dismisses the phase-insensitive assumption, supporting a full matrix formulation that includes cross-interference terms in the SDP analysis.}
\label{fig:protocol}
\end{figure}

The CV-SI-QRNG scheme is illustrated in Fig.~\ref{fig:protocol}.  
To formulate the CV-SI-QRNG security with general POVMs, let $\{M_j\}_{j=1}^m$ represent the general POVM elements obtained from binning a fixed-phase quadrature, such that $0\le M_j\le I$ and $\sum_j M_j = I$. Based on these POVMs, we can establish a finite-dimensional SDP by fixing a Fock cutoff $N$ and defining the orthogonal projectors for the truncated subspace and its complement:

\begin{equation}
\begin{aligned}
P &:= \sum_{n=0}^{N-1}\lvert n\rangle\langle n\rvert, \\
\qquad Q & := I-P.
\end{aligned}
\end{equation}
For any operator $X$, we can thus represent it in a $2\times2$ block form:
\begin{equation}
X=\begin{pmatrix}X_{PP}&X_{PQ}\\ X_{QP}&X_{QQ}\end{pmatrix},
\end{equation}
where $X_{PP}:=PXP$, $X_{PQ}:=PXQ$, and so forth.  With this truncation established, we can systematically define the quantum states alongside their energy constraints. We write the overall normalized state as a decomposition $\rho=\sum_{k=1}^m \rho_k$ with $\rho_k\succeq 0$ , and define the probabilities within and outside the truncated subspace as:
\begin{equation}\label{eq:weights}
\begin{aligned}
a &:= \Tr(P\rho),\\
b &:= \Tr(Q\rho) = 1-a.
\end{aligned}
\end{equation}
We assume that the mean-photon-number is bounded by an observable upper limit, i.e., $\Tr(\rho\,\hat n)\le \mu^\mathrm{U}$. Since the number operator satisfies $\hat n\ge N\,Q$ in the orthogonal complement space, we obtain a rigorous tail bound using Markov's inequality:
\begin{equation}
b = \Tr(Q\rho) \le \frac{\mu^\mathrm{U}}{N}.
\label{eq:tailMarkov}
\end{equation}  To effectively bound the cross-interference terms and analyze the finite-dimensional spectra, we define the maximum eigenvalues of the truncated POVM blocks:
\begin{equation}
s_N := \max_j ||M_{j,PP}||_\infty = \max_j \lambda_{\max}(P M_j P).
\end{equation}
Because $0\le M_j\le I$, we always have $0\le s_N\le 1$. To bound the truncation error rigorously, we rely on the following two foundational lemmas.  
\begin{lemma}\label{lemma:trace-norm}
For any positive semidefinite block operator
\begin{equation}
\rho_k=\begin{pmatrix}\rho_{k,PP}&\rho_{k,PQ}\\ \rho_{k,QP}&\rho_{k,QQ}\end{pmatrix}\succeq 0,
\end{equation}
the trace norm of its off-diagonal block is bounded by:
\begin{equation}
\|\rho_{k,QP}\|_1 \le \sqrt{\Tr(\rho_{k,PP})\,\Tr(\rho_{k,QQ})}.
\label{eq:rhoPQtrace}
\end{equation}
\end{lemma}
\begin{proof}
By Douglas' lemma (or the operator Cauchy--Schwarz factorization for PSD block matrices), there exists a contraction $W_k$ with $\|W_k\|_\infty\le 1$ such that
\begin{equation}
\rho_{k,QP}=\sqrt{\rho_{k,QQ}}\,W_k\,\sqrt{\rho_{k,PP}}.
\end{equation}
Using Hölder's inequality for Schatten norms and unitary invariance, we have:
\begin{equation}
\begin{aligned}
\|\rho_{k,QP}\|_1 &\le \|\sqrt{\rho_{k,QQ}}\|_2\,\|W_k\|_\infty\,\|\sqrt{\rho_{k,PP}}\|_2 \\
&\le \|\sqrt{\rho_{k,QQ}}\|_2\,\|\sqrt{\rho_{k,PP}}\|_2.
\end{aligned}
\end{equation}
Since $\|X\|_2^2=\Tr(X^\dagger X)$, we obtain
\begin{equation}
\begin{aligned}
\|\sqrt{\rho_{k,QQ}}\|_2^2 &=\Tr(\rho_{k,QQ})  \\
\|\sqrt{\rho_{k,PP}}\|_2^2&=\Tr(\rho_{k,PP}),
\end{aligned}
\end{equation}
which yields Eq.~\eqref{eq:rhoPQtrace}.
\end{proof}

\begin{lemma}
For any POVM element satisfying $0\le M_j\le I$, we have:
\begin{equation}\label{eq:MPQmult}
\begin{aligned}
\|M_{j,QP}\|_\infty &\le \sqrt{\|M_{j,PP}\|_\infty\,\|M_{j,QQ}\|_\infty} \\
&\le \sqrt{\|M_{j,PP}\|_\infty}.
\end{aligned}
\end{equation}
\end{lemma}
\begin{proof}
We use the standard inequality for $2\times2$ PSD block operators:
\begin{equation}
\begin{pmatrix}A&B\\ B^\dagger&C\end{pmatrix}\succeq 0 \quad\Longrightarrow\quad \|B\|_\infty^2 \le \|A\|_\infty\,\|C\|_\infty,
\end{equation}
which can be derived from Schur complements.
Applying this to $M_j$ with $A=M_{j,PP}$, $B=M_{j,PQ}$, and $C=M_{j,QQ}$ gives
\begin{equation}
\|M_{j,PQ}\|_\infty \le \sqrt{\|M_{j,PP}\|_\infty\,\|M_{j,QQ}\|_\infty}.
\end{equation}
Since $0\le M_{j,QQ}\le Q$, it implies $\|M_{j,QQ}\|_\infty\le 1$, yielding the final inequality in Eq.~\eqref{eq:MPQmult}.
\end{proof}

\begin{corollary}
By summing the result of Lemma~\ref{lemma:trace-norm} over $k$ and applying the Cauchy--Schwarz inequality, we get:
\begin{equation}
\begin{aligned}
\sum_k \|\rho_{k,QP}\|_1 &\le \sum_k \sqrt{\Tr(\rho_{k,PP})\,\Tr(\rho_{k,QQ})} \\
&\le \sqrt{\Big(\sum_k \Tr(\rho_{k,PP})\Big)\Big(\sum_k \Tr(\rho_{k,QQ})\Big)} \\
&= \sqrt{ab}.
\end{aligned}
\end{equation}
Combining this with Eq.~\eqref{eq:MPQmult}, the cross-interference contribution $T_\times:=2\,\Re\sum_k \Tr(\rho_{k,QP}M_{k,PQ})$ is globally bounded by:
\begin{equation}\label{eq:TxBound}
\begin{aligned}
|T_\times| &\le 2\sum_k \|\rho_{k,QP}\|_1\,\|M_{k,PQ}\|_\infty \\
&\le 2\sqrt{s_N}\,\sqrt{ab},
\end{aligned}
\end{equation}   
where $a$ and $b$ are defined in Eq.~\eqref{eq:weights}.
\end{corollary}

With these bounds, the infinite-dimensional guessing probability can be safely decomposed into block-wise contributions:
\begin{equation}\label{eq:obj}
\begin{aligned}
p_{\mathrm{guess}}^{(\infty)} =& \sum_{k=1}^m \Tr(\rho_{k,PP}M_{k,PP}) + \sum_{k=1}^m \Tr(\rho_{k,QQ}M_{k,QQ}) \\
&+ 2\,\Re\sum_{k=1}^m \Tr(\rho_{k,QP}M_{k,PQ})\\
=&\ T_P + T_Q + T_\times \\
\leq&\ T_P + b + 2\sqrt{s_N}\,\sqrt{ab} \\
=&\ \sum_{k=1}^m \Tr(\rho_{k,PP}M_{k,PP}) + 1 - \sum_{k=1}^m \Tr(\rho_{k,PP}) \\
&+ 2\sqrt{s_N}\,\sqrt{ab}.
\end{aligned}
\end{equation}
Recall the tail bound in Eq.~\eqref{eq:tailMarkov}, where $b \le \mu^\mathrm{U}/N$.
Since $a+b=1$, the term $\sqrt{ab}$ reaches its maximum when $b$ takes its upper bound (assuming $\mu^\mathrm{U}/N \le 1/2$).
Therefore, the cross term is bounded by $2\sqrt{s_N} \frac{\sqrt{\mu^\mathrm{U} (N-\mu^\mathrm{U})}}{N}$.
For the linear constraints representing the experimental statistics, we can similarly bound the deviations caused by the truncation:
\begin{equation}\label{eq:holder}
\begin{aligned}
    \Tr\Big(M_j\sum_{k=1}^m\rho_k\Big) =& \Tr\Big(M_{j,PP}\sum_{k=1}^m\rho_{k,PP}\Big)  \\& +\Tr\Big(M_{j,QQ}\sum_{k=1}^m\rho_{k,QQ}\Big) 
    + T_{\times, j} \\
    \leq& \Tr\Big(M_{j,PP}\sum_{k=1}^m\rho_{k,PP}\Big) \\
    &+ \Big\|\sum_{k=1}^m\rho_{k,QQ}\Big\|_1 \|M_{j,QQ}\|_\infty \\
    &+ 2\sqrt{s_j}\frac{\sqrt{\mu^\mathrm{U} (N-\mu^\mathrm{U})}}{N},
\end{aligned}
\end{equation}
where $s_j = \|M_{j,PP}\|_\infty$ is the maximum eigenvalue of the specific truncated POVM element $M_{j,PP}$ and $T_{\times, j}:= 2 \Re[M_{j,PQ} \sum_k \Tr(\rho_{k,QP})]$. The cross term $T_{\times, j}$ can be bounded similar to $T_{\times}$ by replacing $s_N$ with $s_j$.
By applying the lower and upper bounds of $T_{\times, j}$ and bounding the $QQ$ term via $\mu^\mathrm{U}/N$, we can safely confine the infinite-dimensional statistics within a finite-dimensional constraint band.
Finally, by incorporating all these analytical bounds, we can formulate a strictly finite-dimensional SDP problem that reliably evaluates the randomness:
\begin{equation}\label{eq:sdp_pri_finite}
\begin{aligned}
\max_{\{\rho_{k,PP}\}} \quad & \sum_{k=1}^m \Tr(\rho_{k,PP} M_{k,PP}) + 1 - \sum_{k=1}^m \Tr(\rho_{k,PP}) \\
&+ 2\sqrt{s_N} \frac{\sqrt{\mu^\mathrm{U} (N-\mu^\mathrm{U})}}{N} \\
\mathrm{s.t.} \quad & \Tr\Big(M_{j,PP}\sum_{k=1}^m\rho_{k,PP}\Big) \ge p_j - \frac{\mu^\mathrm{U}}{N} \| M_{j,QQ}\|_\infty \\
&\quad - 2\sqrt{s_j}\frac{\sqrt{\mu^\mathrm{U} (N-\mu^\mathrm{U})}}{N}, \quad \forall j \\
& \Tr\Big(M_{j,PP}\sum_{k=1}^m\rho_{k,PP}\Big) \le p_j \\
&\quad + 2\sqrt{s_j}\frac{\sqrt{\mu^\mathrm{U} (N-\mu^\mathrm{U})}}{N}, \quad \forall j \\
& \Tr\Big(\sum_{k=1}^m \rho_{k, PP}\Big) \le 1, \\
& \Tr\left(\hat{n}\sum_{k=1}^m \rho_{k, PP}\right) \leq \mu^\mathrm{U}, \\
& \rho_{k,PP} \succeq 0, \quad \forall k.
\end{aligned}
\end{equation}
Note that the upper bound in the first constraint does not require the $\|M_{j,QQ}\|_\infty$ term, as $\Tr(M_{j,QQ}\rho_{QQ}) \ge 0$ can only increase the expectation value.

\section{Experimental Implementation}
In this section, we provide a concrete example to demonstrate the application of our method: the POVM is chosen to be a single-quadrature measurement via unbalanced homodyne detection (UBHD).
To strictly bound the randomness away from zero, the maximum eigenvalues of the truncated POVM elements must be bounded away from unity ($s_j < 1$).
Ideal balanced homodyne detection implies perfect projective measurements onto highly squeezed quadrature eigenstates, inevitably leading to $s_j \to 1$, which reduces the min-entropy to zero.
Hence, we exploit a realistic and physically advantageous scheme: UBHD coupled with realistic photodiode inefficiency.
Suppose the unknown signal state interferes with a strong local oscillator (LO) at a beam splitter with an asymmetric transmittance $T \neq 0.5$.
Furthermore, we abandon the idealized lossless assumption and model the physical scenario where both commercial photodiodes exhibit an identical finite quantum efficiency, $\eta < 1$.
In this realistic configuration, the physical loss $\eta$ is modeled by introducing virtual beam splitters on the detection arms, coupling the modes to independent vacuum states. We assume Eve cannot access the internal vacuum ports of Alice's detection apparatus, i.e., she cannot possess the purification of these lost photons. Consequently, this trusted loss cannot be exploited by Eve to gather additional information. Instead, it acts as a secure noise mechanism that physically convolves with the measurement operators. Intuitively, the beam splitter unbalancing penalty together with the trusted loss form a Gaussian convolution which blurs the ideally sharp quadrature projectors into overlapping, unsharp distributions. This physically induced fuzziness prevents the POVM elements from acting as mutually exclusive pure projections, thereby compressing their maximum eigenvalues below unity.

When the commercial detector performs a direct 1:1 subtraction of the photocurrents, the interference between the macroscopic LO amplitudes and the signal modes yields a dominant noise contribution.
The total variance of these combined, independent noise operators is mathematically equivalent to an effective total system efficiency, $\eta_{\mathrm{sys}}$.
By equating the derived noise variance to a pure loss channel, the system efficiency is revealed to be the strict product of the physical detector efficiency and the beam splitter unbalancing penalty:
\begin{equation}\label{eq:totalloss}
\eta_{\mathrm{sys}} = \eta \times 4T(1-T).
\end{equation}
And the matrix form of the discrete POVM elements is given by
\begin{equation}\label{eq:POVM_UHD}
    \begin{aligned}
       M_j =& \frac{1}{2} \int_{-\infty}^{\infty} dy \left[ \text{erf}\left( \frac{x_{j+1} - \sqrt{\eta_{sys}}y}{\sqrt{2(1-\eta_{sys})V_0}} \right) \right.\\
       &- \left.\text{erf}\left( \frac{x_j - \sqrt{\eta_{sys}}y}{\sqrt{2(1-\eta_{sys})V_0}} \right) \right] |y\rangle\langle y|,
    \end{aligned}
\end{equation}
where we take $V_0 =1/2$ as the variance of vacuum fluctuation.
We leave the detailed derivation of Eqs.~\eqref{eq:totalloss} and \eqref{eq:POVM_UHD} in Appendix~\ref{app:POVM_UHD}.

\begin{figure}[htp]
\centering
\includegraphics[width=0.5\textwidth]{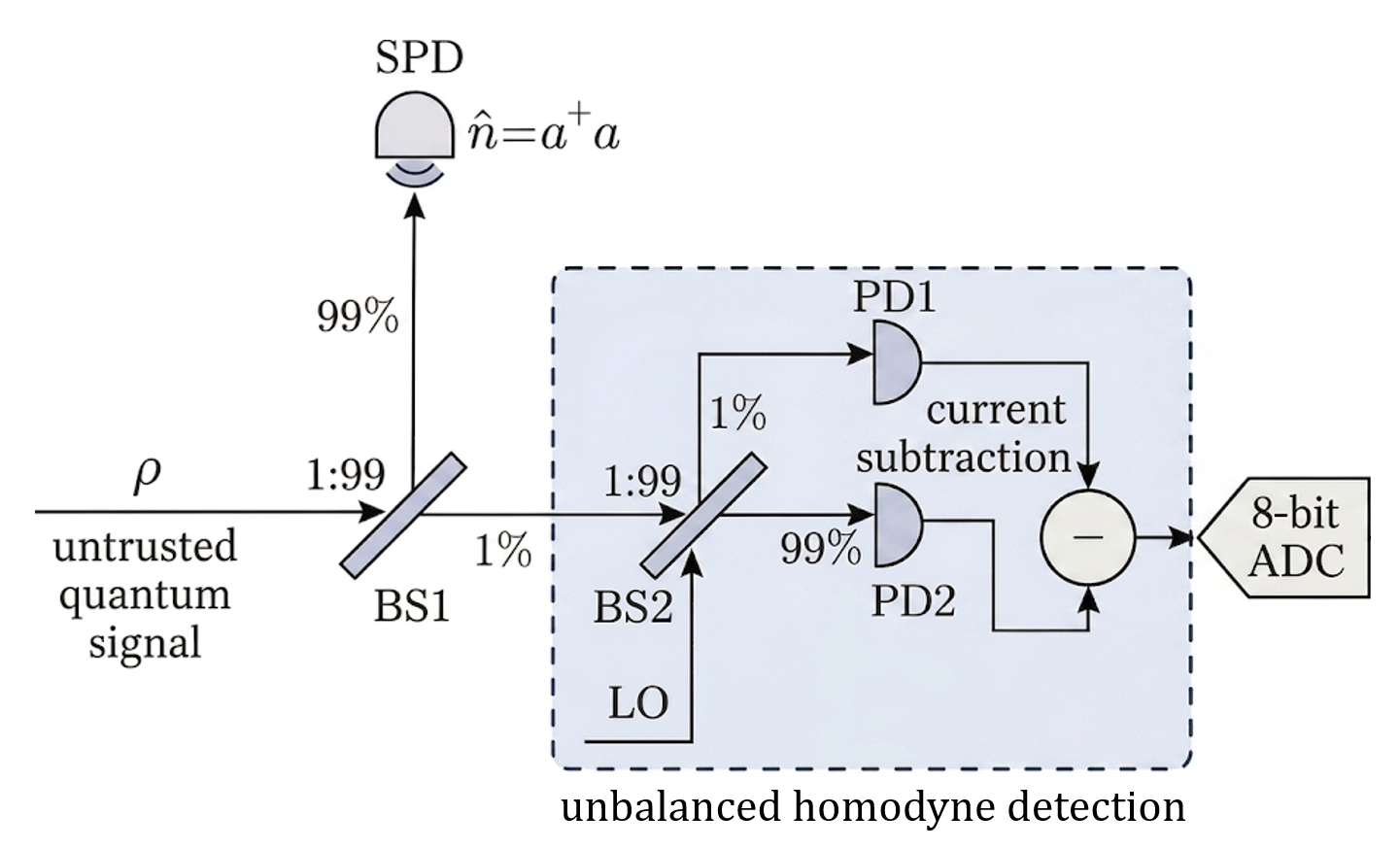}
\caption{ Experimental setup of the CV-SI-QRNG protocol featuring unbalanced homodyne detection (UBHD). The untrusted quantum signal ($\rho$) first passes through an asymmetric beam splitter (BS1) with a 1:99 splitting ratio. The 99\% port is directed to a single-photon detector (SPD) for real-time, highly sensitive photon number monitoring. The remaining 1\% of the signal enters the UBHD module, interfering with a local oscillator (LO) at a second 1:99 beam splitter (BS2). The interference signals are captured by two photodetectors (PD1 and PD2), followed by current subtraction to eliminate common-mode noise. The continuous quadrature outcomes are finally digitized by an 8-bit analog-to-digital converter (ADC) for randomness extraction.}
\label{fig:exp_setting}
\end{figure}


\begin{figure*}[htbp]
    \centering
    \begin{minipage}{0.48\textwidth}
        \centering
        \includegraphics[width=\linewidth]{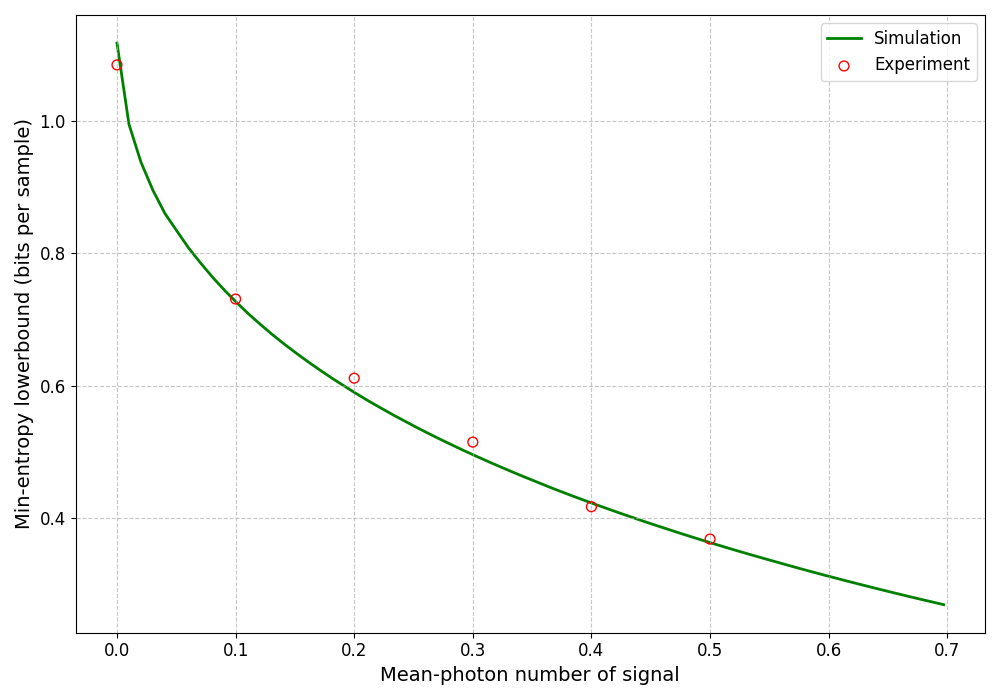} 
    \end{minipage}
    \begin{minipage}{0.48\textwidth}
        \centering
        \includegraphics[width=\linewidth]{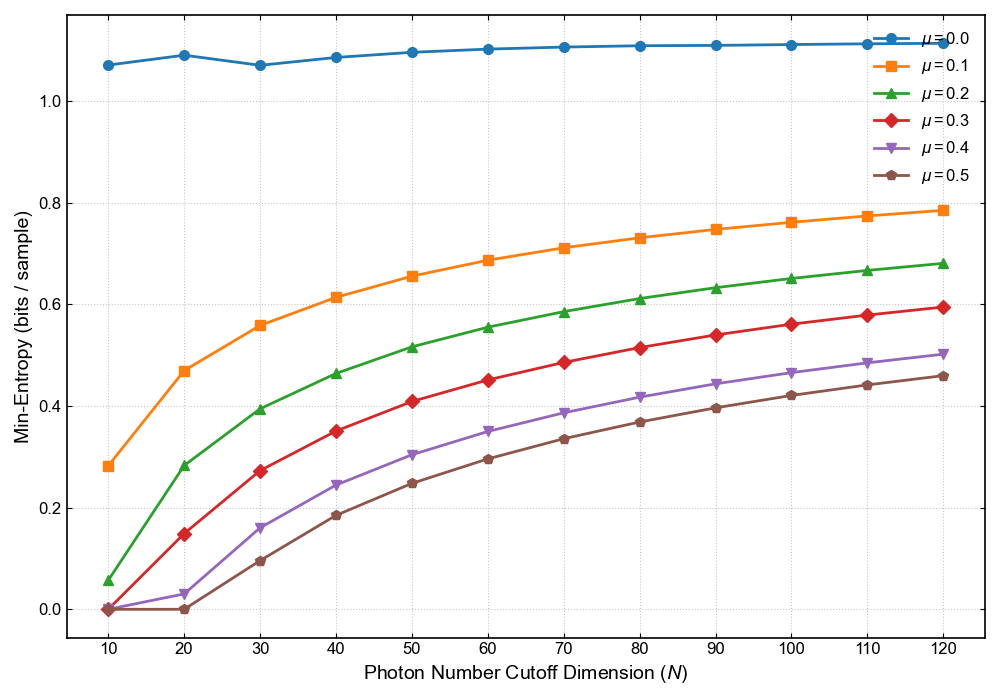}
    \end{minipage}
    \caption{Left: Min-entropy bound versus the mean photon number of the signal ($\mu$). Green line and red points denote the theoretical simulation bounds and the experimental results $(\mu \in \{0,0.1,0.2,0.3,0.4,0.5 \} )$, respectively. We set the photon number cutoff $N=80$ in simulation and experimental cases. Right: Min-entropy bound versus the photon number cutoff for experimental data with different mean photon numbers of signals.}
    \label{fig:sim_vs_exp}
\end{figure*}

In our experimental demonstration, we employ vacuum and weak coherent state light sources with mean photon numbers 0.1, 0.2, 0.3, 0.4, and 0.5 as the untrusted sources. To rigorously bound the maximum photon number entering the randomness generation loop, the source is first routed through a beam splitter with an asymmetric splitting ratio of $1:99$. The $99\%$ port is directed to a single-photon detector (Thorlabs ID cube, low noise version) for real-time intensity monitoring. This monitoring detector features a detection efficiency of $15\%$ and a dark count rate of $3$~kHz. The remaining $1\%$ of the signal enters the unbalanced homodyne detection (UBHD) setup. The optical signals are then measured using a commercial balanced amplified photodetector (Thorlabs PDB480C-AC), which features a responsivity of $0.9$~A/W at $1550$~nm and a wide bandwidth of $30$~kHz to $1.6$~GHz. Prior to randomness generation, a rigorous calibration is performed: the local oscillator (LO) is blocked to accurately measure the electronic noise variance. The pure quantum signal variance is isolated by subtracting this electronic noise variance from the total noise variance, and is subsequently calibrated against standard vacuum fluctuations (assuming the theoretical vacuum variance is $V_0 = 1/2$). The experimental settings are illustrated in Fig.~\ref{fig:exp_setting}.

For the randomness extraction, the continuous quadrature measurement outcomes are discretized into $8$ distinct intervals: $(-\infty, -3.5]$, $[-3.5,-2.33]$, $[-2.33,-1.17]$, $[-1.17,0]$, $[0,1.17]$, $[1.17,2.33]$, $[2.33,3.50]$ and $[3.50,\infty)$. Based on a robust dataset consisting of $2 \times 10^8$ experimental samples for each intensity and utilizing the bounded intensity $\mu^{\mathrm{U}}$ observed in experiments, we calculate the extractable secure randomness to be $1.08$, $0.73$, $0.61$,$0.51$,$0.41$, and $0.36$ bits per sample with $N=80$, corresponding to $\mu = 0, 0.1, 0.2, 0.3, 0.4, 0.5$, respectively (see Fig~\ref{fig:sim_vs_exp}). We also provide a convergence plot showing the min-entropy lower bound varies with the photon number cutoff, achieving a highest randomness of $1.11$ with $\mu=0$ and $N=120$. In the calculation above, we neglect the finite-size effect and take the experimental statistics as the probability value for simplicity. Leveraging the $1.6$~GHz high-speed bandwidth of our detection system, this per-sample entropy enables our continuous-variable source-independent quantum random number generator to reach an ultra-high generation rate of up to $1.776$~Gbps ($1.11 \times 1.6$~Gbps). Raw experimental data are provided in Appendix~\ref{app:exp_data}.

\section{conclusion}
In conclusion, we have established a rigorous numerical security proof framework for CV-SI-QRNGs that accommodates general POVMs. This method faithfully maps the intractable infinite-dimensional SDP into a solvable finite-dimensional space, preserving the information-theoretic security of the guessing probability upper bound. We successfully applied this framework to an unbalanced homodyne detection scheme subjected to physical loss. Our experimental demonstration using vacuum and weak coherent states yields a secure randomness extraction of up to 1.11 bits per sample. This confirms an ultra-high generation rate of 1.776 Gbps at a 1.6 GHz detection bandwidth, demonstrating exceptional robustness even with imperfect hardware.

Looking forward, the versatility of this framework opens new avenues for innovative protocol design. Notably, our approach permits the certification of secure randomness using only a single-quadrature measurement, bypassing the traditional theoretical requirement of measuring two conjugate quadratures (such as $X$ and $P$). For future directions, we expect this critical advantage can further simplify the required experimental architecture and reduces hardware complexity, enabling faster, more robust, and highly practical quantum random number generators. 

\section*{acknowledgment}
H.Z. acknowledges funding from National Natural Science Foundation of China Grants No. 92465202, and the Innovation Funding of ICT, CAS under Grant No. E561150.

\onecolumngrid
\appendix

\section{Derivation of POVMs of a lossy unbalanced homodyne detection}\label{app:POVM_UHD}
In a realistic experimental UBHD setup, the BS transmissivity $T \neq 1/2$. The signal is $\hat{a}$, and the strong LO is $\hat{b} \approx \alpha + \delta\hat{b}$ (with $\alpha = |\alpha|e^{i\theta}$). The output modes are $\hat{c} = \sqrt{T}\hat{a} + \sqrt{1-T}\hat{b}$ and $\hat{d} = \sqrt{1-T}\hat{a} - \sqrt{T}\hat{b}$.
By expanding the photocurrent operators $\hat{I}_c = \hat{c}^\dagger \hat{c}$ and $\hat{I}_d = \hat{d}^\dagger \hat{d}$ and neglecting the extremely weak $\hat{a}^\dagger\hat{a}$ and $\delta\hat{b}^\dagger\delta\hat{b}$ terms, we obtain:
\begin{align}
    \hat{I}_c &\approx (1-T)|\alpha|^2 + (1-T)|\alpha|\hat{x}_{\delta b} + \sqrt{T(1-T)}|\alpha|\hat{x}_\theta \\
    \hat{I}_d &\approx T|\alpha|^2 + T|\alpha|\hat{x}_{\delta b} - \sqrt{T(1-T)}|\alpha|\hat{x}_\theta,
\end{align}
where $\hat{x}_{\delta b}$ is the LO shot noise quadrature, and $\hat{x}_\theta$ is the signal quadrature.
A commercial balanced detector performs direct subtraction: $\Delta \hat{I} = \hat{I}_c - \hat{I}_d$.
\begin{equation}
    \Delta \hat{I} = (1-2T)|\alpha|^2 + (1-2T)|\alpha|\hat{x}_{\delta b} + 2\sqrt{T(1-T)}|\alpha|\hat{x}_\theta.
\end{equation}
In standard operation, an AC-coupling circuit (high-pass filter) blocks the large macroscopic DC term $(1-2T)|\alpha|^2$. The remaining fluctuating AC signal is:
\begin{equation}
    \delta(\Delta \hat{I}) = 2\sqrt{T(1-T)}|\alpha|\hat{x}_\theta + (1-2T)|\alpha|\hat{x}_{\delta b}.
\end{equation}
To obtain the normalized quadrature measurement $\hat{x}_{meas}$, we divide this signal by the gain coefficient of the signal term $2\sqrt{T(1-T)}|\alpha|$:
\begin{equation}
    \hat{x}_{meas} = \hat{x}_\theta + \frac{1-2T}{2\sqrt{T(1-T)}}\hat{x}_{\delta b}.
\end{equation}
The operator $\hat{x}_{meas}$ consists of the ideal signal $\hat{x}_\theta$ corrupted by the un-cancelled LO shot noise. Since $\hat{x}_{\delta b}$ is a standard vacuum mode with variance $V_0$, the effective noise variance added to the measurement is:
\begin{equation}
    V_{LO\_noise} = \left( \frac{1-2T}{2\sqrt{T(1-T)}} \right)^2 V_0 = \frac{(1-2T)^2}{4T(1-T)} V_0 = \frac{1 - 4T + 4T^2}{4T(1-T)} V_0 .
\end{equation}
Let us define an effective detection efficiency $\eta_{UB}$. A pure loss channel $\eta$ introduces a noise variance of $\frac{1-\eta}{\eta}V_0$. Equating these:
\begin{equation}
    \frac{1-\eta_{UB}}{\eta_{UB}} = \frac{1 - 4T(1-T)}{4T(1-T)},
\end{equation}
which implies
\begin{equation}
   \eta_{UB} = 4T(1-T).
\end{equation}

Thus, realistic UBHD with direct subtraction is mathematically perfectly equivalent to an ideal balanced homodyne detector suffering a transmission loss of $\eta_{UB} = 4T(1-T)$,
\begin{equation} \label{eq:POVM2}
    M_j^{\text{ideal}} = \frac{1}{2} \int_{-\infty}^{\infty} dy \left[ \text{erf}\left( \frac{x_{j+1} - \sqrt{\eta_{UB}}y}{\sqrt{2(1-\eta_{UB})V_0}} \right) - \text{erf}\left( \frac{x_j - \sqrt{\eta_{UB}}y}{\sqrt{2(1-\eta_{UB})V_0}} \right) \right] |y\rangle\langle y|.
\end{equation}

Now we assume both photodiodes have an identical physical quantum efficiency $\eta$. We model this by placing virtual BSs (transmissivity $\eta$) on modes $\hat{c}$ and $\hat{d}$, coupling them to independent vacuum modes $\hat{v}_1$ and $\hat{v}_2$. 
The dominant noise contribution from the physical loss arises from the interference between these vacuum modes and the macroscopic classical amplitudes of the LO.
The classical amplitude in arm $c$ is $\sqrt{1-T}\alpha$, and in arm $d$ is $-\sqrt{T}\alpha$. The lossy photocurrents are:
\begin{align}
    \hat{I}_c' &\approx \eta \hat{I}_c + \sqrt{\eta(1-\eta)}\sqrt{1-T}|\alpha|\hat{x}_{v1} \\
    \hat{I}_d' &\approx \eta \hat{I}_d - \sqrt{\eta(1-\eta)}\sqrt{T}|\alpha|\hat{x}_{v2}.
\end{align}
Performing direct subtraction $\Delta \hat{I}' = \hat{I}_c' - \hat{I}_d'$, applying AC-coupling to remove the DC term, and substituting the expanded expressions for $\hat{I}_c$ and $\hat{I}_d$, we get:
\begin{equation}
    \delta(\Delta \hat{I}') = \eta \left[ 2\sqrt{T(1-T)}|\alpha|\hat{x}_\theta + (1-2T)|\alpha|\hat{x}_{\delta b} \right] + \sqrt{\eta(1-\eta)}|\alpha| \left( \sqrt{1-T}\hat{x}_{v1} + \sqrt{T}\hat{x}_{v2} \right).
\end{equation}
We normalize the measurement by dividing by the overall signal gain coefficient $2\eta\sqrt{T(1-T)}|\alpha|$:
\begin{equation}
    \hat{x}_{meas}' = \hat{x}_\theta + \frac{1-2T}{2\sqrt{T(1-T)}}\hat{x}_{\delta b} + \frac{\sqrt{\eta(1-\eta)}}{2\eta\sqrt{T(1-T)}} \left( \sqrt{1-T}\hat{x}_{v1} + \sqrt{T}\hat{x}_{v2} \right).
\end{equation}

The noise modes $\hat{x}_{\delta b}$, $\hat{x}_{v1}$, and $\hat{x}_{v2}$ are mutually independent vacuums, each with variance $V_0$. The total variance of the combined noise operators is the sum of their individual variances.
The variance of the $(\sqrt{1-T}\hat{x}_{v1} + \sqrt{T}\hat{x}_{v2})$ term is $(1-T)V_0 + TV_0 = V_0$. Thus:
\begin{align}
    V_{total\_noise} &= \left[ \frac{(1-2T)^2}{4T(1-T)} + \frac{\eta(1-\eta)}{4\eta^2 T(1-T)} \right] V_0 \nonumber \\
    &= \left[ \frac{\eta(1-2T)^2 + (1-\eta)}{4\eta T(1-T)} \right] V_0 \nonumber \\
    &= \left[ \frac{\eta - 4\eta T + 4\eta T^2 + 1 - \eta}{4\eta T(1-T)} \right] V_0 \nonumber \\
    &= \frac{1 - 4\eta T(1-T)}{4\eta T(1-T)} V_0.
\end{align}
Let us define the total effective system efficiency $\eta_{sys}$. Using the pure loss variance relation $V = \frac{1-\eta_{sys}}{\eta_{sys}}V_0$:
\begin{equation}
    \frac{1-\eta_{sys}}{\eta_{sys}} = \frac{1 - \eta [4T(1-T)]}{\eta [4T(1-T)]}
\end{equation}
This immediately reveals that the total system efficiency is the product of the physical efficiency and the unbalancing penalty:
\begin{equation}
    \eta_{sys} = \eta \times 4T(1-T) = \eta \cdot \eta_{UB}.
\end{equation}
Therefore, the discrete POVM element for a realistic UBHD with direct subtraction and physical detector loss $\eta$ is:
\begin{equation} \label{eq:POVM3}
    M_j = \frac{1}{2} \int_{-\infty}^{\infty} dy \left[ \text{erf}\left( \frac{x_{j+1} - \sqrt{\eta_{sys}}y}{\sqrt{2(1-\eta_{sys})V_0}} \right) - \text{erf}\left( \frac{x_j - \sqrt{\eta_{sys}}y}{\sqrt{2(1-\eta_{sys})V_0}} \right) \right] |y\rangle\langle y|.
\end{equation}

\section{Simulation formulae}\label{app:simulation_foemulae}
In this section, we give the simulation formulae of conditional probabilities $p_j$ and $\mu^\mathrm{U}$ involved in the SDP problem. Compared with balanced homodyne detection, the measurement outcome of unbalanced homodyne detection follows a broadened Gaussian distribution. For an input coherent state $\ket{\sqrt{\mu}}$ utilized for source characterization, the actual measurement yields a Gaussian distribution with its effective mean shifted by the total system efficiency, while maintaining the standard vacuum variance:
\begin{equation}
   P(x) = \frac{1}{\sqrt{\pi}} \exp\left[ - \left(x - \sqrt{2\eta_{\mathrm{sys}}\mu}\right)^2 \right].
\end{equation}
The corresponding experimental statistics $p_j$ utilized for the finite-dimensional SDP constraints in Eq.~\eqref{eq:sdp_pri_finite} are directly given by integrating this distribution:
\begin{equation}
p_k = \int_{x^{(k)}}^{x^{(k+1)}} P(x) dx.
\end{equation}

Then we estimate $\mu^\mathrm{U}$. In a practical implementation, the inherent imperfections of Single-Photon Detectors (SPDs), specifically dark counts, must be taken into account when estimating the upper bound of the mean photon number. 
Let $p_d$ denote the dark count probability per measurement gate and $\eta$ denote the overall detection efficiency. The probability of the SPD not clicking within a single gate, considering both the input signal with mean photon number $\mu$ and the dark counts, is given by $P_{\text{no\_click}} = e^{-\eta\mu} (1 - p_d)$. Since the dark count probability is extremely small ($p_d \ll 1$), this can be well approximated by:
\begin{equation}
    P_{\text{no\_click}} \approx e^{-\eta\mu} e^{-p_d} = e^{-(\eta\mu + p_d)}.
\end{equation}
If we model this combined effect as an equivalent ideal detector responding to a stronger virtual coherent state with an effective mean photon number $\mu^\mathrm{U}$, we have $e^{-\eta\mu^\mathrm{U}} = e^{-(\eta\mu + p_d)}$, which implies:
\begin{equation}
    \mu^\mathrm{U}= \mu + \frac{p_d}{\eta}.
\end{equation}
Therefore, the dark counts introduce an equivalent excess mean photon number $\mu_{d} = p_d / \eta$. This excess noise applies to both the vacuum state ($\mu^\mathrm{U} = \mu_{d}$) and the weak coherent state ($\mu^\mathrm{U} = \mu + \mu_{d}$).

In our experiment, the SPD has a detection efficiency of $\eta = 15\%$. Based on the detector's characteristics and the measurement time window (resulting from a raw dark count rate of $3\text{ kHz}$), the dark count probability per gate is characterized as $p_d = 6 \times 10^{-5}$. Substituting these experimental parameters into the model, the equivalent maximum excess mean photon number introduced by the detector's dark counts is estimated to be:
\begin{equation}
    \mu_{d} = \frac{6 \times 10^{-5}}{0.15} = 4 \times 10^{-4}.
\end{equation}
This excess value is rigorously incorporated into the estimation of $\mu^\mathrm{U}$ to guarantee the source-independent security of the generated random numbers against any potential manipulation of the background noise by an eavesdropper.

By incorporating these analytical matrix elements alongside the observable energy constraint $\mathrm{Tr}(\hat{n}\rho) \le \mu^\mathrm{U}$ which effectively bounds the high-energy tails, we can efficiently and securely compute the randomness lower bound via the proposed SDP framework, accurately reflecting the physical limitations of the optoelectronic components.

\section{Experimental Data}\label{app:exp_data}

In our experiment, the raw quadrature data is collected using an 8-bit analog-to-digital converter (ADC), which originally yields $2^8 = 256$ discrete outcome values. To efficiently extract secure randomness while maintaining a computationally tractable SDP formulation, we apply a classical coarse-graining post-processing step to group these 256 discrete values into $m=8$ bins. 

In our experiment, we choose uniformly distributed boundaries, $(-\infty, -3.5]$, $[-3.5,-2.33]$, $[-2.33,-1.17]$, $[-1.17,0]$, $[0,1.17]$, $[1.17,2.33]$, $[2.33,3.50]$ and $[3.50,\infty)$, for simplicity. Actually, our framework allows arbitrary divisions to effectively squash the maximum eigenvalues of the POVM elements, which strictly restricts the eavesdropper's maximum guessing probability. Since this coarse-graining is performed entirely in the classical data post-processing stage, the implementation of non-uniform intervals is perfectly valid and introduces no security loopholes. We leave the optimization of non-uniform division for future work.


The detailed probability distributions---including the experimental probability, theoretical probability, and absolute error---across the 8 optimized bins for different input signal intensities ($\mu \in \{0, 0.1, 0.2, 0.3, 0.4, 0.5\}$) are summarized in Table~\ref{tab:raw_data}.

\begin{table*}[htp!]
\centering
\caption{Probability distributions of the quadrature measurement outcomes across 8 discrete intervals for different mean photon numbers ($\mu$).}
\label{tab:raw_data}
\renewcommand{\arraystretch}{1.2}
\begin{tabular}{c c | c c c | c c c}
\hline\hline
\multirow{2}{*}{\textbf{Bin}} & \multirow{2}{*}{\textbf{Interval}} & \multicolumn{3}{c|}{$\bm{\mu = 0}$} & \multicolumn{3}{c}{$\bm{\mu = 0.1}$} \\
\cline{3-8}
 & & Exp. Prob. & Theo. Prob. & Abs. Error & Exp. Prob. & Theo. Prob. & Abs. Error \\
\hline
1 & $[-\infty, -3.50]$ & $0.000000$ & $0.000000$ & $3.72 \times 10^{-7}$ & $0.000004$ & $0.000000$ & $4.10 \times 10^{-6}$ \\
2 & $[-3.50, -2.33]$   & $0.000660$ & $0.000483$ & $1.77 \times 10^{-4}$ & $0.000535$ & $0.000289$ & $2.45 \times 10^{-4}$ \\
3 & $[-2.33, -1.17]$   & $0.048600$ & $0.048996$ & $3.96 \times 10^{-4}$ & $0.039098$ & $0.036330$ & $2.77 \times 10^{-3}$ \\
4 & $[-1.17, 0.00]$    & $0.450340$ & $0.450520$ & $1.80 \times 10^{-4}$ & $0.406580$ & $0.407149$ & $5.69 \times 10^{-4}$ \\
5 & $[0.00, 1.17]$     & $0.449000$ & $0.450520$ & $1.52 \times 10^{-3}$ & $0.487333$ & $0.490517$ & $3.18 \times 10^{-3}$ \\
6 & $[1.17, 2.33]$     & $0.050420$ & $0.048996$ & $1.42 \times 10^{-3}$ & $0.065112$ & $0.064921$ & $1.91 \times 10^{-4}$ \\
7 & $[2.33, 3.50]$     & $0.000960$ & $0.000483$ & $4.77 \times 10^{-4}$ & $0.001333$ & $0.000792$ & $5.40 \times 10^{-4}$ \\
8 & $[3.50, \infty]$   & $0.000020$ & $0.000000$ & $1.96 \times 10^{-5}$ & $0.000005$ & $0.000001$ & $4.71 \times 10^{-6}$ \\
\hline\hline

\multirow{2}{*}{\textbf{Bin}} & \multirow{2}{*}{\textbf{Interval}} & \multicolumn{3}{c|}{$\bm{\mu = 0.2}$} & \multicolumn{3}{c}{$\bm{\mu = 0.3}$} \\
\cline{3-8}
 & & Exp. Prob. & Theo. Prob. & Abs. Error & Exp. Prob. & Theo. Prob. & Abs. Error \\
\hline
1 & $[-\infty, -3.50]$ & $0.000002$ & $0.000000$ & $1.44 \times 10^{-6}$ & $0.000001$ & $0.000000$ & $1.28 \times 10^{-6}$ \\
2 & $[-3.50, -2.33]$   & $0.000434$ & $0.000170$ & $2.64 \times 10^{-4}$ & $0.000415$ & $0.000098$ & $3.17 \times 10^{-4}$ \\
3 & $[-2.33, -1.17]$   & $0.035557$ & $0.026463$ & $9.09 \times 10^{-3}$ & $0.031233$ & $0.018933$ & $1.23 \times 10^{-2}$ \\
4 & $[-1.17, 0.00]$    & $0.389619$ & $0.362016$ & $2.76 \times 10^{-2}$ & $0.371497$ & $0.316655$ & $5.48 \times 10^{-2}$ \\
5 & $[0.00, 1.17]$     & $0.508284$ & $0.525550$ & $1.73 \times 10^{-2}$ & $0.518492$ & $0.554149$ & $3.57 \times 10^{-2}$ \\
6 & $[1.17, 2.33]$     & $0.064860$ & $0.084525$ & $1.97 \times 10^{-2}$ & $0.076853$ & $0.108148$ & $3.13 \times 10^{-2}$ \\
7 & $[2.33, 3.50]$     & $0.001240$ & $0.001275$ & $3.54 \times 10^{-5}$ & $0.001502$ & $0.002013$ & $5.11 \times 10^{-4}$ \\
8 & $[3.50, \infty]$   & $0.000005$ & $0.000002$ & $3.04 \times 10^{-6}$ & $0.000005$ & $0.000003$ & $1.61 \times 10^{-6}$ \\
\hline\hline

\multirow{2}{*}{\textbf{Bin}} & \multirow{2}{*}{\textbf{Interval}} & \multicolumn{3}{c|}{$\bm{\mu = 0.4}$} & \multicolumn{3}{c}{$\bm{\mu = 0.5}$} \\
\cline{3-8}
 & & Exp. Prob. & Theo. Prob. & Abs. Error & Exp. Prob. & Theo. Prob. & Abs. Error \\
\hline
1 & $[-\infty, -3.50]$ & $0.000001$ & $0.000000$ & $1.08 \times 10^{-6}$ & $0.000001$ & $0.000000$ & $1.14 \times 10^{-6}$ \\
2 & $[-3.50, -2.33]$   & $0.000371$ & $0.000055$ & $3.16 \times 10^{-4}$ & $0.000358$ & $0.000031$ & $3.27 \times 10^{-4}$ \\
3 & $[-2.33, -1.17]$   & $0.027956$ & $0.013304$ & $1.47 \times 10^{-2}$ & $0.027948$ & $0.009180$ & $1.88 \times 10^{-2}$ \\
4 & $[-1.17, 0.00]$    & $0.375430$ & $0.272444$ & $1.03 \times 10^{-1}$ & $0.360664$ & $0.230539$ & $1.30 \times 10^{-1}$ \\
5 & $[0.00, 1.17]$     & $0.514539$ & $0.575065$ & $6.05 \times 10^{-2}$ & $0.527303$ & $0.587361$ & $6.01 \times 10^{-2}$ \\
6 & $[1.17, 2.33]$     & $0.080134$ & $0.136004$ & $5.59 \times 10^{-2}$ & $0.082102$ & $0.168128$ & $8.60 \times 10^{-2}$ \\
7 & $[2.33, 3.50]$     & $0.001563$ & $0.003121$ & $1.56 \times 10^{-3}$ & $0.001618$ & $0.004750$ & $3.13 \times 10^{-3}$ \\
8 & $[3.50, \infty]$   & $0.000005$ & $0.000006$ & $1.08 \times 10^{-6}$ & $0.000007$ & $0.000011$ & $4.47 \times 10^{-6}$ \\
\hline\hline
\end{tabular}
\end{table*}

\bibliography{bibsemidiqrng}

\end{document}